%% file: main.tex
\documentclass[11pt]{article}

\usepackage{amsmath,amsthm,amssymb,amsfonts}
\usepackage{enumitem}
\usepackage{aliascnt}
\usepackage[colorlinks]{hyperref}
\hypersetup{linkcolor=blue,citecolor=blue}
\usepackage{cleveref}
\usepackage{fullpage}
\usepackage{tcolorbox}
\usepackage{complexity}

\usepackage{pgfplots}
\pgfplotsset{compat=1.18}
\usepackage{tikz}
\usetikzlibrary{quantikz2}

\allowdisplaybreaks
\hypersetup{
  pdftitle={Exactification and Discrete-Basis Robustness for QAC0},
  pdfauthor={Daniel Grier, Jackson Morris, and Kewen Wu}
}

\Crefformat{equation}{(#2#1#3)}

\newtheorem{theorem}{Theorem}
\newaliascnt{lemma}{theorem}
\newtheorem{lemma}[lemma]{Lemma}
\aliascntresetthe{lemma}
\newaliascnt{proposition}{theorem}

\aliascntresetthe{proposition}
\newaliascnt{corollary}{theorem}
\newtheorem{corollary}[corollary]{Corollary}
\aliascntresetthe{corollary}
\newaliascnt{claim}{theorem}

\aliascntresetthe{claim}
\newaliascnt{fact}{theorem}

\aliascntresetthe{fact}
\newaliascnt{definition}{theorem}
\newtheorem{definition}[definition]{Definition}
\aliascntresetthe{definition}
\newaliascnt{question}{theorem}

\aliascntresetthe{question}
\newaliascnt{remark}{theorem}

\aliascntresetthe{remark}

\crefname{theorem}{theorem}{theorems}
\crefname{lemma}{lemma}{lemmas}
\crefname{proposition}{proposition}{propositions}
\crefname{corollary}{corollary}{corollaries}
\crefname{claim}{claim}{claims}
\crefname{fact}{fact}{facts}
\crefname{definition}{definition}{definitions}
\crefname{question}{question}{questions}
\crefname{remark}{remark}{remarks}

\providecommand{\QAC}{\ComplexityFont{QAC}}
\providecommand{\BQAC}{\ComplexityFont{BQAC}}
\providecommand{\EQAC}{\ComplexityFont{EQAC}}

\providecommand{\QNC}{\ComplexityFont{QNC}}

\providecommand{\CZ}{\mathsf{CZ}}
\providecommand{\AND}{\mathsf{AND}}
\providecommand{\OR}{\mathsf{OR}}

\providecommand{\EX}{\mathsf{EX}}

\providecommand{\THR}{\mathsf{THR}}
\providecommand{\Fanout}{\mathsf F}
\providecommand{\Fanout}{\mathsf{CZ}}

\providecommand{\poly}{\operatorname{poly}}
\providecommand{\polylog}{\operatorname{polylog}}

\providecommand{\eps}{\varepsilon}
\providecommand{\ket}[1]{\lvert #1\rangle}
\providecommand{\bra}[1]{\langle #1\rvert}
\providecommand{\braket}[2]{\langle #1\mid #2\rangle}

\providecommand{\proj}[1]{\lvert #1\rangle\!\langle #1\rvert}
\providecommand{\sket}[1]{\,|#1\rangle}
\providecommand{\norm}[1]{\left\lVert #1\right\rVert}
\providecommand{\abs}[1]{\left\lvert #1\right\rvert}
\providecommand{\set}[1]{\left\{#1\right\}}

\newcommand{\PREP}{\mathrm{PREP}}
\newcommand{\SEL}{\mathrm{SEL}}
\newcommand{\err}{\mathsf{err}}
\newcommand{\init}{\mathsf{init}}
\newcommand{\target}{\mathsf{targ}}

\title{The Robustness of $\QAC^0$}
\author{
Daniel Grier\thanks{UCSD. Email: \texttt{dgrier@ucsd.edu}}
\and Jackson Morris\thanks{UCSD. Email: \texttt{jrm035@ucsd.edu}}
\and Kewen Wu\thanks{Caltech. Email: \texttt{shlw\_kevin@hotmail.com}.}
}
\date{}

\begin{document}
\maketitle
\begin{abstract}
In this work we study the robustness of $\QAC^0$ with respect to error tolerance and modifications to its gate-set. First, we investigate whether the non-zero error typically allowed for $\QAC^0$ circuits computing Boolean functions is truly necessary. We show that the error inherent in the parallel $W$-test of \cite{grier_morris_wu} can be eliminated entirely via a novel application of exact amplitude amplification in the many-copies context. Consequently, we find that $\QAC^0$ can \emph{exactly} simulate $\TC^0$ with polynomially many copies of the classical input and that for every fixed prime $p$ exact $\QAC^0$, $\EQAC^0$, can compute total Boolean functions outside of $\AC^0[p]$.

Second, we ask to what extent the computational power of $\QAC^0$ follows from the fact that arbitrary single-qubit gates may be used at any point in the circuit. We find that $\QAC^0$ is in fact robust to restrictions on which single-qubit gates are permitted: every $\QAC^0$ circuit can be approximately implemented by a $\QAC^0$ circuit consisting of just generalized Toffoli, $S$, and Hadamard gates. Moreover, this approximating circuit can be constructed efficiently from a classical description of the original circuit.
\end{abstract}

\tableofcontents

\clearpage
\input{sections/introduction}
\input{sections/preliminaries}
\input{sections/discrete_basis}
\input{sections/exactification_1}

\section{AI disclosure}

Results in this paper were obtained through interactions with GPT 5.6 Sol: the authors provided intuitive justification and intended proof approach, then GPT confirmed the heuristics with rigorous proofs, which were further simplified and generalized in the same human-LLM interaction cycle.
GPT 5.6 Sol and GPT 6 were used to draft the paper, which was fully rewritten by the authors.
All mathematical claims and exposition remain the responsibility of the authors.

\bibliographystyle{alphaurl}
\bibliography{bibli}

\appendix
\input{sections/appendices}

\end{document}

%% file: sections/introduction.tex
\section{Introduction}\label{sec:introduction}
Understanding which features of a computational model are essential to its power is a central goal of quantum complexity theory. Two very natural aspects of a computational model which may be essential for its power are 1) the set of allowed operations (or more specifically gates in circuit models) and 2) whether computations in the model can be made errorless. Further, these questions are particularly delicate for constant-depth quantum circuit models as standard methods for error reduction and and gate synthesis may require a super-constant factor blow-up in circuit depth. In this work, we explore these questions in the context of $\QAC^0$ - the class of problems decidable by constant-depth, polynomial-size quantum circuits which use polynomially many ancilla qubits and consist of generalized Toffoli and arbitrary single-qubit gates.

This class was originally introduce by Moore as a quantum analogue of $\AC^0$ in \cite{moore:1999}, though it has proven to be substantially more difficult to characterize than its classical counterpart. While a plethora of combinatorial and analytic techniques have given us a very precise understanding of the limitations of $\AC^0$ \cite{ajtai_parity, fss_ac0, hastad_thesis, lmn}, the same cannot be said for $\QAC^0$ despite numerous lower bounds in restricted settings \cite{depth2, depth3, anshu_dong_ou_yao:2024_QAC0, pauli_spec}. As such, $\QAC^0$ sits at the frontier of our understanding of shallow quantum computation. 

Recent work has shown that $\QAC^0$ can solve tasks beyond the reach of $\AC^0$ - a phenomena that is provably false for the weaker quantum circuit model $\QNC^0$. Explicitly, \cite{grier_morris_wu} showed that bounded-error $\QAC^0$ circuits can compute total Boolean functions outside $\AC^0[p]$ for every fixed prime (p). They also showed that arbitrary $\mathrm{TC}^0$ circuits can be simulated by to inverse exponential error by a $\QAC^0$ circuit supplied with polynomially many copies of the classical input.

One limitation of this result is that inverse-exponential error in either result cannot be eliminated in an obvious way. Another is that the gates used in these constructions are somewhat non-standard; for different input sizes $n$, the quantum circuit witnessing the simulation requires different single-qubit gates. Moreover, these gates do not lie in the Clifford hierarchy in general.

These limitations of the construction raise the following broader questions about $\QAC^0$:
\begin{enumerate}
    \item Must $\QAC^0$ circuits have some non-zero when computing Boolean functions or can any such computation be made exact? Does $\QAC^0$'s advantage against $\AC^0[p]$ vanish if zero-error is required?
    \item Are all single-qubit gates necessary for $\QAC^0$ to retain its power? Does $\QAC^0$ admit a discrete gate-set, i.e., can $\QAC^0$ circuits be implemented with just generalized Toffoli and finitely many single qubit gates?
\end{enumerate}
A number of methods for gate synthesis and error reduction have been developed for generic quantum circuits - however, they cannot be immediately applied to the constant-depth regime. In particular, the Solovay-Kitaev theorem shows that any fixed, inverse-closed set of single-qubit gates generating a dense subgroup of $\mathsf{SU}(2)$ can approximate an arbitrary single-qubit unitary to accuracy $\epsilon$ using only $\polylog(1/\epsilon)$ gates \cite{Solovay1995, Kitaev1997}. However, this method in general result in a compiled circuit whose depth is a super-constant factor larger than the original.

Shi further showed that the Hadamard and Toffoli gate can be leveraged to simulate arbitrary quantum computations, so a gate-set need not generate a dense subset of $\mathsf{SU}(2)$ in order to be universal \cite{shi_universal}. Again though, this construction can in general result in an increase in circuit depth by a factor of $n$, the number of qubits the circuit acts on.

Our first main result is that $\QAC^0$ does indeed admit a simple, discrete gate-set:
\begin{theorem}[See also \Cref{thm:HST_universal}]
\label{thm:intro-discrete-basis}
For every $\epsilon > 0$ and $\QAC$ circuit, $Q$, of depth-$d$ and size-$s$ acting on $n$ qubits there exists a depth-$O(d)$, size $O(s\log^2{s/\epsilon}$) circuit, $\tilde{Q}$ acting on $O( n + s\log^2{s/\epsilon})$ qubits such that $\Tilde{Q}$ only consists of generalized Toffoli, Hadamard, and $S$ gates and has operator-norm distance at most $\epsilon$ from $\mathbb{I}_m \otimes Q$ (where $\mathbb{I}_m$ accounts for ancilla qubits used by $\Tilde{Q}$). Moreover, $\Tilde{Q}$ can be found efficiently given the description of $Q$.
\end{theorem}

Regarding error elimination, exact amplitude amplification provides a method for eliminating error when the initial success probability is \emph{known} \cite{Hoy00, BHMT02}. However, a direct application need not preserve circuit depth: the standard construction uses a sequence of applications of the original circuit and its inverse and the precise number of iterations depends on the original success probability. For the $W$-test, the rejection probability on an incorrect input can be as small as $\Theta(1/n^2)$, leading to a non-constant number of sequential amplification steps. Parallel repetition improves this probability to a constant, but its precise value still depends on the Hamming weight of the input, which is precisely the quantity we are interested in computing. Obtaining an exact constant-depth implementation therefore requires overcoming both the sequential nature of the amplification and its dependence on the initial success probability. 

Our second main result does precisely this:
\begin{theorem}[See also \Cref{thm:zero_error_ex}]
\label{thm:intro-exact-tc0}
Every language in $\TC^0$ can be computed exactly by a family of constant-depth, polynomial-size $\QAC$ circuits using polynomially many copies of its input. Consequently,
\[
                 \TC^0 \subseteq \EQAC^0\circ\NC^0.
\]
\end{theorem}
Further, this exactification carries over to the decision problem separation of \cite{grier_morris_wu}, i.e., $\EQAC^0 \not \subset \AC^0[p]$ for any fixed prime $p$ (see \Cref{cor:exact_separation}), meaning that $\QAC^0$ can still exhibit quantum advantage against $\AC^0$ even in the zero-error regime. 

In addition to these two robustness results we develop several new $\QAC^0$ primitives which may be of interest for further $\QAC^0$ constructions.

The first primitive is a random selector: given an $n$-bit string $x$, it returns a random coordinate $i\in[n]$ with $x_i=1$. In addition, we show that both its binary and unary encodings can be obtained in $\QAC^0$.

\begin{theorem}[Random selector; see also \Cref{thm:one-hot-selector} and \Cref{thm:binary-selector}]\label{thm:intro-selectors}
There are constant-depth, polynomial-size $\QAC$ circuits that, on nonzero $x\in\{0,1\}^n$, with probability at least $1-2^{-\polylog(n)}$, output a uniformly random location $i\in[n]$ with $x_i=1$, encoded either in binary or as an $n$-bit indicator vector $e_i$.
\end{theorem}

Our second primitive is an approximate counter that estimates the Hamming weight of an input binary string. While it is a combination of known geometric counting sketch \cite{ajtai_approximate_counting,viola_approximate_counting} and known $\QAC^0$ primitives \cite{grier_morris_wu}, we find it very interesting and document it here.

\begin{theorem}[Approximate counter; see also \Cref{thm:approx-counter}]
\label{thm:intro-approx-counter}
There is a constant-depth, polynomial-size $\QAC$ circuit that, on input $x\in\{0,1\}^n$, with probability at least $1-2^{-\polylog(n)}$, computes the Hamming weight of $x$ up to $1/\polylog(n)$ relative errors.
\end{theorem}

\subsection{Technical overview}\label{sec:overview}
Next, we give an overview  of our techniques and highlight the novel aspects of each construction. 

\paragraph{A discrete gate-set for $\QAC^0$.}
Standard quantum compilation techniques may result in a super-constant factor depth blow-up. In particular, synthesis over a fixed set of single-qubit gates generally improves accuracy by using a longer sequence of gates acting on the target qubit. To circumvent this our construction encodes the parameters of a target gate in a coherent sampling circuit whose size increases with the desired precision while its depth remains constant. 

Our starting point is a construction which for any $p \in [0, 1]$ and $\epsilon > 0$ produces a constant-depth circuit on $m  = O(\log^2{1/\epsilon})$ qubits, $Q$, such that 
\begin{itemize}
    \item $Q$ consists only of Hadamard and Generalized Toffoli gates and
    \item Measuring the final qubit of $Q\ket{0^m}$ yields $\ket{1}$ with probability $\tilde{p} \in [p - \epsilon, p + \epsilon]$
\end{itemize}
$Q$ is found via a binary search-like procedure wherein we iteratively construct a DNF whose acceptance probability for a uniformly random input approaches $p$ as more clauses are added. Finally, evaluating this DNF on $\ket{+}^m$ yields the desired single-qubit marginal which will be entangled with the ancilla used by the rest of the circuit.

Next, we make use of the fact that up to a $Z$ correction any real single-qubit gate can be written as
\begin{align*}
    U_{\theta} := \begin{pmatrix}
        \cos\theta & \sin\theta\\
        \sin\theta & -\cos\theta
      \end{pmatrix}
    = \cos\theta\,Z+\sin\theta\,X.
\end{align*}
for some $\theta \in [0, 2\pi]$. Moreover, we may write $U_{\theta}$ as a linear combination of $\{-X, X, -Z, Z\}$ using only non-negative coefficients. Using the coherent samplers and an additional Hadamard gate, we prepare a superposition of labels for these four operators, with probabilities given by the normalized coefficients. Next, we control on the label register to apply $\{-X, X, -Z, Z\}$ coherently, resulting in the application of $\Tilde{U_{\theta}}$ and not exactly $U_{\theta}$ due to the error incurred by the coherent sampler. After uncomputation of the sampler registers, the non-negative coefficients are always normalized in such a way that in the all-zeros ancilla branch $U_{\tilde{\theta}}$ has been applied and this branch has amplitude exactly $\frac{1}{2}$. Finally, we apply amplitude amplification to approximately obtain $\ket{0^m}\Tilde{U_{\theta}}\ket{\psi}$.

Thus, we can approximately implement arbitrary real single-qubit unitaries using Hadamard and generalized Toffoli in constant depth - with the $S$ gate and real single-qubit gates we can construct an approximate implementation of any single-qubit gate in constant depth.

\paragraph{Exactificaiton of the many-copy $W$-test.}
First, recall the techniques developed in \cite{grier_morris_wu} to simulate $\TC^0$ via a $\QAC^0$ circuit which acts on multiple copies of the input string. The primary technical tool underlying this simulation is a primitive called the ``$W$-test". The $W$-test is a $\QAC^0$ circuit which takes $x \in \{0, 1\}^n$ as input and produces a state $\ket{\psi_x}$ of the form
\begin{align*}
    \ket{\psi_x} &= \frac{n - 2|x|}{n}\ket{\psi_x^0}\ket{0} + \sqrt{1 - \bigg{(}\frac{n - 2|x|}{n}\bigg{)}^2}\ket{\psi_x^1}\ket{1}
\end{align*}
where $\ket{\psi_x^0}$ and $\ket{\psi_x^1}$ are some normalized states. The measurement outcome of this final qubit is treated as the output of the $W$-test. This means that
\begin{itemize}
    \item If $|x| = n/2$ then measuring the final qubit of $\ket{\psi_x}$ yields $\ket{1}$ with probability $1$
    \item If $|x| \not = n/2$ then measuring the final qubit of $\ket{\psi_x}$ yields $\ket{0}$ with probability at least $1/n^2$
\end{itemize}
Hence, the $W$-test computes $\EX_{n/2}$ with perfect completeness and $1 - 1/\poly(n)$ soundness. Running the $W$-test in parallel $m$ times will yield outcomes $b_1, \dots b_m \in \{0, 1\}$ and if $|x| = n/2$ then all of these bits will be $1$ with certainty. However, if $|x| \not = n/2$ then at least one $b_i$ will be equal to $0$ with probability $\bigg{(}1 - \frac{1}{n^2} \bigg{)}^{m}$. So in this case, if $m \gg n^2\log{n}$, then $\mathbb{P}[b_1 \land \cdots \land b_m = 1] \leq 1/n$. When $m$ copies of $x$ are given as input the $W$-test can be ran in parallel and then $\AND$ of their outputs yields $\EX_{n/2}(x)$ with high probability. Indeed, the error can be made inverse exponential with $n^{2 + \Omega(1)}$ copies of the input. However, standard techniques (i.e. parallel repetition) will fail to improve the error beyond this.

To go beyond inverse-exponential error and indeed remove it completely, we turn to exact amplitude amplification. Explicitly, our goal to take $\ket{\psi_x}^{\otimes m}$ and amplify the amplitude on the branch in which \emph{some} $W$-test yields $\ket{0}$, if this occurs with non-zero probability. In the case of $|x| = n/2$ nothing need be done, as the qubit corresponding to the value of $b_1\land \cdots \land b_m$ is already $\ket{1}$. If $|x| \not = n/2$ then the desired branch has weight $\Omega(1)$ when $m = \Omega(n^2)$, but precisely how much $\ket{\psi_x}^{\otimes m}$ overlaps with the desired state depends on $|x|$ - the quantity this construction is meant to reveal some information about. 

To address this apparent circularity we observe that if $|x|$ is promised to be either $n/2$ or some $k \in \{0, 1, \dots n\}\setminus \{n/2\}$ then $\ket{\psi_x}^{\otimes m}$'s outcome qubits will either be $\ket{1}^{\otimes m}$ or have some fixed constant amplitude on the branch in which \emph{some} output qubit is $\ket{0}$. Hence, for this promise problem exact amplitude amplification can be applied in constant depth to perfectly distinguish strings of weight $k$ from strings of weight $n/2$. With sufficiently many classical copies of $x$ we may solve each of these promise problems in parallel and in the case that $|x| = k \not = n/2$, the corresponding promise problem $W$-test will reveal this with certainty.
\subsection{Future directions}\label{sec:conclusion}
One immediate question left open by our work is whether the $S$ gate is still necessary for $\QAC^0$. Of course, any $\QAC$ circuit implemented with just $H$ and generalized Toffoli will result in a unitary with real entries. One reason to suggest that this gate may not be necessary for $\QAC^0$ is that no $\QAC^0$ construction for computing a Boolean function uses this gate and it has been proven to be unnecessary once the circuit depth can grow like $\Omega(n)$ \cite{aharonov_universal}.

Additionally, the fact that arbitrary $\QAC^0$ circuits can be compiled with this discrete gate-set may allow for new lower-bound approaches which take advantage of additional structure this new gate-set entails.

Similarly, the construction presented in \Cref{sec:exactification} combined with prior work \cite{grier_morris_wu} shows that every Boolean function known to be in $\BQAC^0$ is also computable in $\EQAC^0$. So, we leave open the question of whether $\BQAC^0 = \EQAC^0$. It is quite unclear how a generic exactification procedure might be implemented, given that it is still quite unclear which types of Boolean functions are even in $\BQAC^0$.

%% file: sections/preliminaries.tex
\section{Preliminaries}\label{sec:prelim}

For integer $n\ge1$, we use $[n]$ to denote the set $\{1,2,\ldots,n\}$.
We use $\mathbb{I}_n$ to denote the identity operator on $n$ qubits, and omit $n$ when the dimension is clear.

For every $x\in\{0,1\}^n$, we use $|x|$ to denote its Hamming weight.
For binary strings $x,y$ of equal length, we use $x\oplus y$ to denote their bitwise XOR.
For integers $k\ge0$, $\EX_k(x)$ and $\THR_{\ge k}(x)$ indicate whether $|x|=k$ and $|x|\ge k$, respectively.

Our circuit upper bounds (i.e., constructions) are \emph{uniform} circuits that can be efficiently extracted from our proofs and descriptions. 
All our circuit lower bounds hold with respect to \emph{non-uniform} circuits.
The (non-)uniformity is not our focus and we do not discuss it in detail.

\paragraph{Classical circuit complexity.}
In a Boolean circuit, every gate evaluates some Boolean function on the input wires and passes the result onward via its output wires.
For a Boolean circuit, 
\begin{itemize} 
\item \emph{depth} is defined as the maximal length from an input bit to output and \emph{size} is defined as the total number of gates in the circuit;
\item \emph{fanin} is the maximal number of input wires of each gate, \emph{fanout} is the maximal number of output wires of each gate, and \emph{gate set} is the set of different gates.
\end{itemize}
We use $\neg$ to denote the negation gate; and use $\AND$ (resp., $\OR$) to denote the AND (resp., OR) gate/function. Sometimes we use $\AND_n,\OR_n$ to highlight that the function takes $n$ bit input and it will be clear from the context when we omit it.

We will need the following standard classical circuit classes. Here we give informal description and refer readers to textbooks \cite{arora2009computational,jukna2012boolean} for formal definitions.
\begin{itemize}
\item $\NC^0$ is the set of languages that can be exactly decided by constant-depth polynomial-size Boolean circuits with bounded fanin and unrestricted fanout. We also use $\NC^0$ for maps with polynomially many output bits computed by such circuits.
\item $\AC^0$ is the set of languages that can be exactly decided by constant-depth polynomial-size Boolean circuits of gate set $\{\neg,\AND,\OR\}$.
\item $\AC^0[m]$ is the set of languages that can be exactly decided by constant-depth polynomial-size Boolean circuits of gate set $\{\neg,\AND,\OR,\MOD_m\}$.
$\ACC^0$ is the union of $\AC^0[m]$ for all $m\ge2$.
\item $\TC^0$ is the set of languages that can be exactly decided by constant-depth polynomial-size Boolean circuits of gate set $\{\neg,(\THR_{\ge k})_{k\ge0}\}$.
\end{itemize}
For decision problems, the following strict inclusions are known to be true for every fixed prime $p$: $\AC^0\subsetneq\AC^0[p]\subsetneq\TC^0$  \cite{ajtai_parity,fss_ac0,hastad_thesis,razborov,smolensky}. It is also easy to prove that $\ACC^0 \subseteq \TC^0$. While it is widely believed that this inclusion is strict, it is yet to be proven.

\paragraph{(Multi-qubit) quantum gates.}
The quantum fanout gate $\Fanout_n$ is an $(n+1)$-qubit unitary defined by
$$
\Fanout_n\colon\ket b\ket{x_1,\ldots,x_n}\to\ket b\ket{x_1\oplus b,\ldots,x_n\oplus b}
\quad\text{for every $x_1,\ldots,x_n,b\in\{0,1\}$.}
$$
Every Boolean function $f\colon\{0,1\}^n\to\{0,1\}$ naturally induces an $(n+1)$-qubit unitary $U_f$ by
$$
U_f\colon\ket{x_1,\ldots,x_n}\ket b\to\ket{x_1,\ldots,x_n}\ket{b\oplus f(x_1,\ldots,x_n)}
\quad\text{for every $x_1,\ldots,x_n,b\in\{0,1\}$.}
$$
The generalized Toffoli gate corresponds to $U_{\AND_n}$ for all $n\ge1$.
When clear from context, we sometimes use $\AND,\OR$ to denote $U_\AND,U_\OR$.

\paragraph{Quantum circuits.}
A quantum circuit $C$ is a product of layers of non-overlapping quantum gates, i.e., a single layer consists of quantum gates acting on disjoint sets of qubits and a circuit consists of multiple layers of gates:
\begin{itemize}
\item A circuit $C$ has \emph{depth} $d$ if $C=M_dM_{d-1}\cdots M_1$ where each $M_i$ is a product of gates operating on disjoint sets of qubits. The \emph{size} of $C$ is the total number of gates in $C$ and qubits that $C$ operates on.
\item Let $S$ be a set of unitaries. We say $C$ has \emph{gate set} $S$ if every gate of $C$ belongs to $S$, i.e. $C$ only consists of gates belonging to $S$.
\end{itemize}

Let $f\colon\{0,1\}^n\to\{0,1\}$.
We say $C$ decides $f$ with $a$ ancillas and error $\epsilon$ if for every $x\in\{0,1\}^n$, the last qubit of $C\ket{x}\ket{0^a}\ket0$ measures, in the computational basis, to $\ket{f(x)}$ with probability at least $1-\epsilon$.
In addition, we say 
\begin{itemize} 
\item it has completeness $c$ if for every $x\in f^{-1}(1)$, the output qubit measures to $\ket{1}$ with probability at least $c$; 
\item and has soundness $s$ if for every $x\in f^{-1}(0)$, the output qubit measures to $\ket{1}$ with probability at most $s$.
\end{itemize}
If $C$ decides $f$ with zero error, then we say $C$ exactly decides $f$ and we can assume without loss of generality $C\ket x\ket{0^a}\ket0=\ket x\ket{0^a}\ket{f(x)}$ by standard uncomputation, which incurs an insignificant constant blowup in the depth and size of the circuit.

\paragraph{$\QAC$ and $\QAC^0$.}
We primarily work with $\QAC$ circuits and the corresponding quantum circuit class $\QAC^0$.
See \cite{nielsen2010quantum,bera2010quantum} for a more comprehensive introduction.

A quantum circuit is a $\QAC$ circuit if it only uses single-qubit gates and generalized Toffoli gates.
We distinguish $\QAC^0$ as the following finer classes $\EQAC^0$ and $\BQAC^0$.
\begin{itemize}
\item $\EQAC^0$ is the set of languages that can be exactly decided by constant-depth polynomial-size $\QAC$ circuits.
\item $\BQAC^0$ is the set of languages that can be decided with error\footnote{This $1/3$ is not essential and can be boosted to arbitrarily small constant in a black-box way. If necessary, we sometimes also give more precise error bound in terms of completeness and soundness.} at most $1/3$ by constant-depth polynomial-size $\QAC$ circuits.
\end{itemize}

We emphasize that the standard definition of $\QAC^0$ allows for \emph{arbitrary} single-qubit gates, i.e., the gateset includes the set of all single-qubit gate . We allow ancillary qubits in our quantum circuits, the number of which will be upper bounded by the circuit size. We also remark that the quantum $\OR$ gate $U_\OR\colon\ket{x}\ket b\to\ket x\ket{b\oplus\OR(x)}$ is in $\QAC^0$.

\paragraph{Classical copies and composition.}
We write $\EQAC^0\circ\NC^0$ for Boolean function families of the form $f=g\circ h$, where $h$ is an $\NC^0$ map with polynomial output length and $g$ is computable in $\EQAC^0$. In our applications, $h$ is the copying map $x\mapsto(x,\ldots,x)$, which supplies polynomially many classical copies of $x$.
For a Boolean function $f\colon\{0,1\}^n\to\{0,1\}$ and an integer $m\ge1$, define its $m$-copy lift $f^{\uparrow m}\colon(\{0,1\}^n)^m\to\{0,1\}$ by
\[
    f^{\uparrow m}(x^{(1)},\ldots,x^{(m)})=
    \begin{cases}
        f(x^{(1)}) & \text{if }x^{(1)}=\cdots=x^{(m)},\\
        0 & \text{otherwise.}
    \end{cases}
\]
Thus $f^{\uparrow m}$ is a total Boolean function on $mn$ input bits, with inconsistent copies rejected.

\paragraph{Unitary approximation and block encodings.}
We use $\norm{\cdot}$ for the Euclidean norm of vectors and the operator norm of operators.
In this work we use the following notion of approximation when discussing a quantum circuit which approximates a unitary:
\begin{definition}[Approximation with ancillas]\label{def:approximation-with-ancillas}
Let $U$ be a unitary on $n$ qubits and let $C$ be a circuit using $n + a$ qubits. The circuit $C$ approximates $U$ to error $\eps$ if
\[
    \norm{C\ket{0^a}\ket\psi-\ket{0^a}U\ket\psi}\leq\eps
\]
for every $n$-qubit state $\ket\psi$.
\end{definition}

\Cref{def:approximation-with-ancillas} controls both the error on the input register and the error in uncomputing the ancillas.

\begin{definition}[Block encoding]\label{def:block-encoding}
The all-zero ancilla block of a unitary circuit $L$ using $a$ ancillas is the operator $B$ given by
\[
    B:=(\bra{0^a}\otimes\mathbb{I})L(\ket{0^a}\otimes\mathbb{I}).
\]
Equivalently, $B$ is determined by
\[
    L\ket{0^a}\ket\psi=\ket{0^a}B\ket\psi+\ket\Phi,
\]
where $\ket\Phi$ may depend on $\ket\psi$ and is orthogonal to the subspace in which all $a$ ancillas are zero. We call $L$ a \emph{block encoding} of $B$.
\end{definition}

Measuring the ancillas after applying $L$ gives the all-zero outcome with probability $\norm{B\ket\psi}^2$. In particular, if $B=U/2$ for a unitary $U$, this probability is $1/4$, and the conditional state is $U\ket\psi$. Thus, block-encoding $U/2$ does not by itself implement $U$ with ancillas returned to zero. \Cref{sec:discrete-basis} uses amplitude amplification to convert an approximate block encoding of this form into an approximation satisfying \Cref{def:approximation-with-ancillas}.

\paragraph{Prior constructions.}
We will make use of the following $\QAC^0$ constructions from~\cite{grier_morris_wu}.
\begin{itemize}[beginpenalty=10000]
    \item Preparation, with clean ancillas, of the $W$ state, i.e., a unitary $U$ which can be implemented by a constant-depth, polynomial-size $\QAC$ circuit satisfying
\[
        U\ket{0^n}\sket{0^{a(n)}} = \sket{W_n}\sket{0^{a(n)}}=\frac1{\sqrt n}\sum_{i=1}^n\ket{e_i}\sket{0^{a(n)}}
    \]
    for some $a(n) = \poly(n)$ \cite[Theorem 17]{grier_morris_wu}.
    \item Exact implementation of $\Fanout_m$ when $m=\polylog(n)$ \cite[Corollary 10]{grier_morris_wu}.
    \item Exact Computation of $\THR_k$ and $\EX_k$ on $n$ bits when $k\leq\polylog(n)$ \cite[Corollary 16]{grier_morris_wu}.
\end{itemize}

%% file: sections/discrete_basis.tex
\section{A discrete basis for \texorpdfstring{$\QAC^0$}{QAC0}}
\label{sec:discrete-basis}

In this section we show that arbitrary one-qubit gates can be approximately implemented with generalized Toffoli, $S$, and Hadamard gates. Consequently, this yields a discrete gate set for bounded-error $\QAC^0$ circuits. We first prove that arbitrary real one-qubit gates can be approximately implemented in $\QAC^0$ with just generalized Toffoli and Hadamard gates. It is then simple to show how to extend this construction to all complex single-qubit gates with $S$ gates.

First note that up to multiplication by a $Z$ gate (and global phase), every one-qubit unitary with real entries has the form
\begin{equation}
\label{eq:U-theta}
    U_{\theta} := \begin{pmatrix}
        \cos\theta & \sin\theta\\
        \sin\theta & -\cos\theta
      \end{pmatrix}
    = \cos\theta\,Z+\sin\theta\,X.
\end{equation}
In fact, let us suggestively expand the final expression so that it is a linear combination of unitaries (LCU) with non-negative coefficients:
\[
        U_{\theta} = \left(\frac{1 + \cos\theta}{2}\right) (+Z) + \left(\frac{1 - \cos\theta}{2}\right)(-Z) + \left(\frac{1 + \sin\theta}{2}\right)(+X) + \left(\frac{1 - \sin\theta}{2}\right)(-X).
\]
We follow the standard LCU outline:
\begin{enumerate}
\item \textbf{Preparation oracle:} For any bias $p \in [0,1]$, we construct a circuit $W_{p}$ for preparing approximations to $p$-biased states:
\[
    W_p \ket{0 \cdots 0} = \sqrt{1- \widetilde p} \ket{\phi_0}\ket0 + \sqrt{\widetilde p} \ket{\phi_1}\ket1
\]
for $\widetilde p \approx p$. The states $\ket{\phi_0}$ and $\ket{\phi_1}$ are unimportant. The relevant feature is that measuring the last qubit yields $\ket{1}$ with probability approximately $p$. See~\Cref{lem:coherent-bernoulli}. 

These $W_p$ unitaries will be the workhorses of the preparation oracle. Specifically, for $q_Z := \frac{1 - \cos \theta}{2}$ and $q_X := \frac{1 - \sin \theta}{2}$, we show $(H \otimes W_{q_Z} \otimes W_{q_X}) \ket{0 \cdots 0}$ correctly prepares (approximations to) the coefficients in $(1/2)U_\theta$: 
\[
     \sqrt{\frac{1-\widetilde q_Z}{2}} \ket{\Phi_{+Z}}  + \sqrt{\frac{\widetilde q_Z}{2}} \ket{\Phi_{-Z}}
    +\sqrt{\frac{1-\widetilde q_X}{2}} \ket{\Phi_{+X}} + \sqrt{\frac{\widetilde q_X}{2}} \ket{\Phi_{-X}}
\]
for orthonormal states $\ket{\Phi_{+Z}}$, $\ket{\Phi_{-Z}}$, $\ket{\Phi_{+X}}$, and $\ket{\Phi_{-X}}$.

\item \textbf{Selector oracle:} As it turns out, there is a simple constant-depth $\QNC^0$ circuit that applies Pauli $P \in \{\pm Z,\pm X\}$ based on the flag qubits of the states $\ket{\Phi_{P}}$ prepared above. Thus, after applying this ``selector'' oracle, we get
\[
     \sqrt{\frac{1-\widetilde q_Z}{2}} \ket{\Phi_{+Z}} \otimes Z \ket{\psi} - \sqrt{\frac{\widetilde q_Z}{2}} \ket{\Phi_{-Z}} \otimes Z\ket{\psi}
    +\sqrt{\frac{1-\widetilde q_X}{2}} \ket{\Phi_{+X}} \otimes X \ket{\psi} - \sqrt{\frac{\widetilde q_X}{2}} \ket{\Phi_{-X}} \otimes X\ket{\psi}.
\]
After uncomputing the preparation oracle, one can show that we are left with the state
\[
    \frac{1}{2} \ket{0 \cdots 0} \widetilde U_{\theta} \ket{\psi} + \ket{\Phi}
\]
where $\widetilde U_{\theta}$ is our approximation for $U_\theta$ with error determined by the error in $\widetilde q_X$ and $\widetilde q_Z$, and $\ket{\Phi}$ is orthogonal to all-zeros on the first register.
\item \textbf{Oblivious amplitude amplification:} Finally, we amplify the amplitude on $\ket{0 \cdots 0} \widetilde U_{\theta} \ket{\psi}$ so that it approximates $U_\theta$ rather than $(1/2) U_\theta$. To do this, we use a standard approach of negating the branch in which the ancillas are all zero in order to reflect about $\ket\Phi$. Since the angle between $\ket\Phi$ and $(1/2) \ket{0 \cdots 0} \widetilde U_{\theta} \ket{\psi} + \ket{\Phi}$ is approximately $\pi/6$, only a single round of amplitude amplification is required, leading to a constant-depth implementation.
\end{enumerate}

It's worth emphasizing that each step of the proof introduces little ancilla overhead (specifically, $O(\log^2(1/\epsilon))$ additional ancillas for an $\epsilon$-approximation), and only uses Hadamard, $S$, and generalized Toffoli gates. Steps 2 and 3 of this construction are a fairly straightforward application of LCU + amplitude amplification, but the coherent sampler construction does not seem to have appeared in prior literature, so let us start with that.

\begin{lemma}[Coherent Bernoulli sampler]
\label{lem:coherent-bernoulli}
For every $p\in[0,1]$ and $0<\delta<1/2$, there is a circuit
$W_p$ consisting only of Hadamard and generalized Toffoli gates such that
\begin{equation}
\label{eq:coherent-bernoulli}
    W_p\ket{0^r}\ket{0}
      = \sqrt{1-\widetilde p}\,\ket{\phi_0}\ket{0}
        +\sqrt{\widetilde p}\,\ket{\phi_1}\ket{1},
    \qquad
    \abs{\widetilde p-p}\leq\delta,
\end{equation}
where $r=O(\log^2(1/\delta))$. Furthermore, $W_p$ has
$O(\log^2(1/\delta))$ gates and depth independent of $p$ and $\delta$.
\end{lemma}

\begin{proof}
The first goal of this proof will be to express $1-p$ as an OR of ANDs. If we can do this, then the quantum circuit $W_p$ satisfying \eqref{eq:coherent-bernoulli} is
\begin{center}

\newcommand{\cvdots}{\raisebox{6pt}{$\vdots$}}

\newcommand{\ortarg}{\gate[style={draw, circle, inner sep=-2pt, minimum size=3mm}]{\vee}}
\begin{quantikz}[row sep={0.62cm,between origins}, column sep=0.55cm]
  % ---------------- register r_1 ----------------
  \lstick[4]{$\ket{0^{r_1}}$} & \gate{H} & \ctrl{4} &          &                        \\
                    & \gate{H} & \ctrl{3} &          &                        \\
  \setwiretype{n}   & \cvdots  &   &          &                        \\
                    & \gate{H} & \ctrl{1} &          &                        \\
  \lstick{$\ket{0}$}&          & \targ{}  & \ctrl{7} &         \\
  % ------------- the remaining registers --------
  \setwiretype{n}   & \cvdots  & \cvdots  &          &                        \\
  % ---------------- register r_k ----------------
  \lstick[4]{$\ket{0^{r_k}}$} & \gate{H} & \ctrl{4} &          &                        \\
                    & \gate{H} & \ctrl{3} &          &                        \\
  \setwiretype{n}   & \cvdots  &  &          &                        \\
                    & \gate{H} & \ctrl{1} &          &                        \\
  \lstick{$\ket{0}$}&          & \targ{}  & \ctrl{1} &          \\
  % ---------------- output qubit ----------------
  \lstick{$\ket{0}$}&          & \gate{X} & \ortarg  & 
\end{quantikz}
\end{center}
where measuring the last qubit of this circuit yields $\ket{1}$ with probability 
\[
    \widetilde{p} = \prod_{j = 1}^k \left(1 - \frac{1}{2^{r_j}} \right) 
\]
Therefore, to prove the lemma, it suffices to choose the $r_i$ values so that $|\widetilde p - p| \le \delta$. Our choices will correspond to a sort of modified binary search algorithm, where we will pick the $r_j$ values so that $\prod_{j = 1}^k (1 - 2^{-r_j})$ approximates $p$ from above. That is, for the $i$th iteration of our algorithm, let 
\[
    \widetilde p_i := \prod_{j = 1}^i \left(1 - \frac{1}{2^{r_j}} \right)
\]
be our current estimate for $p$ with the guarantee that $\widetilde p_i \ge p$. We will measure our progress at each step by the gap
\[
    g_i := \widetilde p_i - p.
\]
To start, we set $\widetilde p_0 = 1$ with gap $g_0 = \widetilde p_0 - p = 1-p$. On iteration $i > 0$, pick the smallest integer $r_i > 0$ such that 
\[
    p \le \left(1 - \frac{1}{2^{r_i}} \right) \widetilde p_{i-1}.
\]
In other words, we shrink our estimate as much as possible without making it smaller than our target $p$. We claim this shrinks the gap by at least a factor of $2$. To see this, first note that the gap decreases by
\begin{equation}
\label{eq:gap_diff}
    g_{i-1} - g_i = \widetilde p_{i-1} - \widetilde p_i = \widetilde p_{i-1} - \left(1 - \frac{1}{2^{r_i}} \right) \widetilde p_{i-1} = \frac{1}{2^{r_i}}  \widetilde p_{i-1}.
\end{equation}
Meanwhile, we know that decreasing $r_i$ by one would make the estimate less than $p$. Thus, the gap $g_i$ is at most the difference of those two possible estimates: 
\begin{equation}
\label{eq:gap_bound}
    g_i \le \left(1- \frac{1}{2^{r_i}}\right) \widetilde p_{i-1} -\left(1- \frac{1}{2^{r_i-1}}\right) \widetilde p_{i-1} = \frac{1}{2^{r_i}} \widetilde p_{i-1}.
\end{equation}
Combining \eqref{eq:gap_diff} and \eqref{eq:gap_bound}, we get that $g_i \le (1/2) g_{i-1}$. We stop as soon as $g_i \le \delta$, so $O( \log(1/\delta))$ iterations suffice to $\delta$-approximate $p$. Finally, notice that combining \eqref{eq:gap_diff} and \eqref{eq:gap_bound}, we also get
\[
    g_{i-1} \le \frac{2}{2^{r_i}}\widetilde p_{i-1} \le \frac{2}{2^{r_i}}
\]
so $r_i < 1 + \log_2(1/\delta)$; otherwise, $g_{i-1} \le \delta$ and we would have stopped the algorithm at the previous iteration. Thus, each ancilla register $r_i$ has size $O(\log(1/\delta))$, leading to the final ancilla count $\sum_i (r_i+1) = O(\log^2(1/\delta))$.
\end{proof}
We now show how the previous construction can be used as the first step in the LCU synthesis of $U_{\theta}$:

\begin{lemma}[Direct synthesis of $U_\theta$]
\label{lem:real_approx}
For every $\theta\in\mathbb R$ and $\epsilon \in (0, 1/2)$ there exists a constant-depth circuit $C_\theta$ which acts on $n+1$ qubits, where $n = O(\log^2(1/\epsilon))$, and uses only Hadamard and generalized Toffoli gates such that
\begin{align*}
 \| C_{\theta}\ket{0^n}\ket{\psi} - \ket{0^n}U_{\theta}\ket{\psi}\| \leq \epsilon
\end{align*}
for all single-qubit states $\ket{\psi}$.
\end{lemma}
\begin{proof}
Set
\begin{equation}
\label{eq:two-biases}
       q_Z=\frac{1-\cos\theta}{2},
       \qquad
       q_X=\frac{1-\sin\theta}{2},
       \qquad
       \delta=\frac{\epsilon}{16}.
\end{equation}
We first apply \Cref{lem:coherent-bernoulli} twice in parallel to prepare the two states
\begin{align*}
    \ket{\psi_Z} &:= W_{q_Z} \ket{0 \cdots 0, 0} = \sqrt{1-\widetilde q_Z} \ket{\alpha_0,0} + \sqrt{\widetilde q_Z} \ket{\alpha_1,1} \\
    \ket{\psi_X} &:= W_{q_X} \ket{0 \cdots 0, 0} = \sqrt{1-\widetilde q_X} \ket{\beta_0,0} + \sqrt{\widetilde q_X} \ket{\beta_1,1},
\end{align*}
where $\abs{\widetilde q_Z-q_Z}\leq\delta$ and $\abs{\widetilde q_X-q_X}\leq\delta$. Note that we have written $\ket{0 \cdots 0}$ to avoid the messiness of specifying the exact number of ancilla qubits (a number $r = O(\log^2 (1/\delta))$ in this case). We will continue to use this convention throughout.

Together with Hadamard, these two unitaries, $W_{q_Z}$ and $W_{q_X}$, constitute the preparation circuit in LCU. To see this, first recall our expansion of our target unitary $U_\theta$:
\[
        U_{\theta} = \left(\frac{1 + \cos\theta}{2}\right) (+Z) + \left(\frac{1 - \cos\theta}{2}\right)(-Z) + \left(\frac{1 + \sin\theta}{2}\right)(+X) + \left(\frac{1 - \sin\theta}{2}\right)(-X).
\]
Now, expand $\ket{+} \ket{\psi_Z} \ket{\psi_X}$ as
\[
     \sqrt{\frac{1-\widetilde q_Z}{2}} \ket 0 \ket{\alpha_0, 0}\ket{\psi_X} + \sqrt{\frac{\widetilde q_Z}{2}} \ket{0} \ket{\alpha_1, 1} \ket{\psi_X} 
    +\sqrt{\frac{1-\widetilde q_X}{2}} \ket 1 \ket{\psi_Z} \ket{\beta_0, 0} + \sqrt{\frac{\widetilde q_X}{2}} \ket{1} \ket{\psi_Z} \ket{\beta_1, 1}.
\]
In other words, the squared-magnitudes of the states in this decomposition are approximately those of our decomposition of $U_\theta$ (up to the attenuating $1/2$ factor). At $\delta=0$, this correspondence is exact.

Next, we apply the appropriate unitary (either $+Z$, $-Z$, $+X$, or $-X$) to our target state depending on the basis state to mirror the decomposition of $U_\theta$. This is simple to do with some constant-sized controlled-unitary operations (see \Cref{fig:LCU}). The effect of this selection unitary $\SEL$ on $\ket{+} \ket{\psi_Z} \ket{\psi_X} \otimes I$ can then be written as
\[
     \sqrt{\frac{1-\widetilde q_Z}{2}} \ket{\Phi_{+Z}} \otimes Z - \sqrt{\frac{\widetilde q_Z}{2}} \ket{\Phi_{-Z}} \otimes Z
    +\sqrt{\frac{1-\widetilde q_X}{2}} \ket{\Phi_{+X}} \otimes X  - \sqrt{\frac{\widetilde q_X}{2}} \ket{\Phi_{-X}} \otimes X.
\]
where we've renamed the basis states in our decomposition since their details are now unimportant.

From here, we follow the standard outline for LCU with amplitude amplification. Let's now show the remainder of these steps. First, we uncompute the preparation circuit, i.e., apply $\PREP^\dag$. Altogether, our LCU circuit $L$ is
\[
    L = (\PREP^\dag \otimes I) \circ \SEL \circ (\PREP \otimes I)
\]
where the full circuit is shown in \Cref{fig:LCU}. Note that $L$ is Hermitian since the individual terms of $\SEL$ are Hermitian.
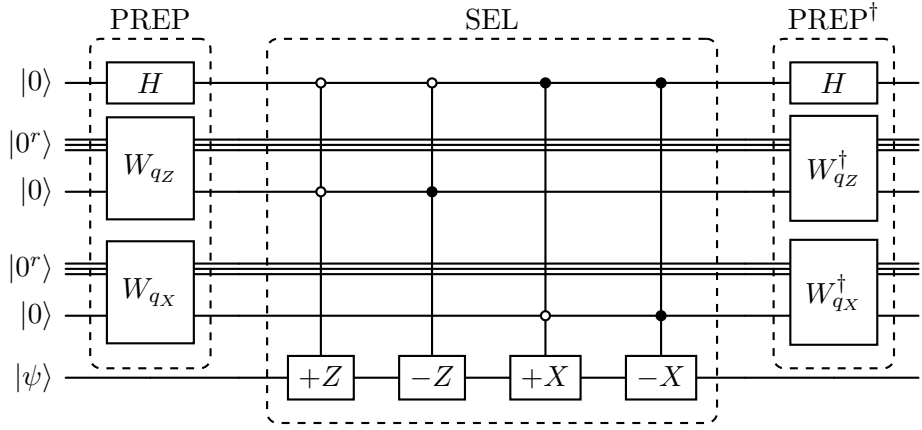
\begin{figure}
\begin{center}
\newcommand{\blockgap}{0.1cm}
% ---- shared sizes -------------------------------------------------------
\newcommand{\gapcell}{\phantomgate{\hspace{\blockgap}}}  
\tikzset{
  prep/.style   = {minimum width=1.15cm},                          % common width
  prepw/.style  = {minimum width=1.15cm, minimum height=1.35cm},   % + common height
  ggbox/.style  = {dashed, rounded corners, inner xsep=4pt, inner ysep=5pt}
}

\begin{quantikz}[row sep={0.62cm,between origins}, column sep=0.6cm,
                 classical gap=0.07cm]
  \lstick{$\ket{0}$}
    & \gate[style=prep]{H}\gategroup[5,steps=1,style=ggbox]{$\PREP$}
    & \gapcell
    & \octrl{2}\gategroup[6,steps=4,style=ggbox]{$\SEL$}
                    & \octrl{2}  & \ctrl{4}   & \ctrl{4}
    & \gapcell
    & \gate[style=prep]{H}\gategroup[5,steps=1,style=ggbox]{$\PREP^{\dagger}$}
    &                            \\[2mm]
  % -------- first register: n bundled qubits + the z bit --------
  \lstick{$\ket{0^r}$}\setwiretype{b}
    & \gate[2,style=prepw]{W_{q_Z}}
    &               &            &            &            &
    &               & \gate[2,style=prepw]{W_{q_Z}^{\dagger}}
    &                            \\
  \lstick{$\ket{0}$}
    &               &            & \octrl{3}  & \ctrl{3}   &            &
    &               &            &            \\[4mm]
  % -------- second register: n bundled qubits + the x bit --------
  \lstick{$\ket{0^r}$}\setwiretype{b}
    & \gate[2,style=prepw]{W_{q_X}}
    &               &            &            &            &
    &               & \gate[2,style=prepw]{W_{q_X}^{\dagger}}
    &                            \\
  \lstick{$\ket{0}$}
    &               &            &            &            & \octrl{1}  & \ctrl{1}
    &               &            &            \\[2mm]
  % -------- output bit --------
  \lstick{$\ket{\psi}$}
    &               &            & \gate{+Z}  & \gate{-Z}  & \gate{+X}  & \gate{-X}
    &               &            &
\end{quantikz}
\end{center}
\caption{The LCU circuit $L = (\PREP^\dag \otimes I) \circ \SEL \circ (\PREP \otimes I)$ for $(1/2)U_\theta$.}
\label{fig:LCU}
\end{figure}

Projecting all-but-the-last qubit onto the all-zeros state, we get
\[
    \widetilde{V}_\theta := (\bra{0 \cdots 0} \otimes I) L (\ket{0 \cdots 0} \otimes I) = 
    \frac{1-\widetilde q_Z}{2} (+Z) + \frac{\widetilde q_Z}{2}(-Z) 
    + \frac{1-\widetilde q_X}{2}(+X) + \frac{\widetilde q_X}{2}(-X) 
\]
where $\widetilde{V}_\theta$ is our approximation for $(1/2) U_\theta$:
\begin{equation}
\label{eq:V_theta_approx}
 \left\|\widetilde{V}_\theta - (1/2)U_\theta\right\|
 \leq \abs{\widetilde q_Z-q_Z}+\abs{\widetilde q_X-q_X}
 \leq 2\delta
 =\frac{\epsilon}{8}.
\end{equation}
For obvious reasons, we will define our approximation for $U_\theta$ itself as $\widetilde U_\theta := 2 \widetilde V_\theta$. In other words, for every single-qubit state $\ket\psi$ we have
\begin{equation}
\label{eq:L_action}
     L \ket{0 \cdots 0} \ket\psi = \frac{1}{2} \ket{0 \cdots 0} \widetilde U_{\theta} \ket{\psi} + \ket{\Phi}
\end{equation}
where $\ket{\Phi}$ is an unnormalized state orthogonal to all-zeros on the preparation register, i.e., $(\bra{0 \cdots 0} \otimes I)\ket{\Phi} = 0$.

In fact, in order to analyze the effect of amplitude amplification, it will be useful to expand $L \ket{0 \cdots 0} \ket\psi$ to include an error term capturing the distance from our target unitary $U_\theta$:
\begin{equation}
\label{eq:L_expanded}
    L \ket{0 \cdots 0} \ket\psi = \frac{1}{2} \ket{0 \cdots 0} U_\theta\ket{\psi} + \ket{\err} + \ket\Phi
\end{equation}
where $\ket{\err} := \frac{1}{2} \ket{0 \cdots 0} (\widetilde U_{\theta} - U_\theta)\ket{\psi}$ is an error term (which depends on $\psi$) such that $\norm{\ket{\err}} \le \epsilon/8$ by \eqref{eq:V_theta_approx}.

The two reflections we will use for amplitude amplification are 
\begin{enumerate}
\item Reflection about $\ket\Phi$: $R := (I - 2\proj{0 \cdots 0})\otimes I$. Note that $R$ can be constructed from a single generalized Toffoli gate. This amplification is ``oblivious'' in the sense that we don't know the state $\ket\Phi$ but can nevertheless reflect about it.
\item Reflection about $L \ket{0 \cdots 0} \ket{\psi}$: $-L R L^\dag$. Recall that $L$ is Hermitian, so $-L R L$ suffices. The global phase can be prepared using $(HX)^4 = - I$.
\end{enumerate}

Given the large amplitude on $\ket{0 \cdots 0} \widetilde U_{\theta} \ket{\psi}$ for our state $L\ket{0 \cdots 0} \ket{\psi}$, it will turn out that only a single round of amplification is required. In other words, our entire circuit $C$ is: state preparation ($L$); then reflection about the orthogonal branch ($R$); then reflection about our prepared state ($-L R L$). 

Altogether, we have $C = - LRLRL$. Algebraically, we can analyze $C$ as follows. First, let's write $\Pi := \proj{0 \cdots 0} \otimes I$, so that $R = I - 2 \Pi$. 

Expanding the $R$ reflections in our circuit $C$, we get
\[
    C = - LRLRL = - L(I - 2 \Pi)L(I - 2 \Pi)L = 2\Pi L + 2L\Pi - L - 4L\Pi L\Pi L.
\]
Define $\ket\init := \ket{0 \cdots 0} \ket{\psi}$ and $\ket\target := \ket{0 \cdots 0} U_\theta \ket{\psi}$, and apply \eqref{eq:L_expanded} to $\ket\init$ and $\ket\target$, respectively, to get
\begin{align}
    \Pi L \ket{\init} &= \frac{1}{2} \ket{0 \cdots 0} U_\theta \ket{\psi} + \ket{\err_1}  = \frac{1}{2} \ket\target + \ket{\err_1} \\
    \Pi L \ket{\target} &= \frac{1}{2} \ket{0 \cdots 0} \ket{\psi}  + \ket{\err_2} = \frac{1}{2} \ket{\init} + \ket{\err_2} \label{eq:U_cancel}
\end{align}
Crucially, \eqref{eq:U_cancel} uses the fact that $U_\theta^2 = I$ (i.e., $U_\theta$ is both unitary and Hermitian). We now have all the pieces we need. Applying $C$ to our initial state, we get
\begin{align*}
     C \ket{\init} &= 2\Pi L\ket\init + L \ket\init - 4L\Pi L\Pi L \ket\init \\
        &= 2( (1/2)\ket\target + \ket{\err_1}) + L \ket\init - 4L\Pi L ( (1/2)\ket\target + \ket{\err_1}) \\
        &= \ket\target + L \ket\init - 2 L\Pi L \ket\target  - 4 L\Pi L \ket{\err_1} + 2\ket{\err_1} \\
        &= \ket\target + L \ket\init - 2 L ((1/2) \ket{\init} + \ket{\err_2})  - 4 L\Pi L \ket{\err_1} + 2\ket{\err_1}\\ 
        &=\ket\target- 2L\ket{\err_2}  - 4 L\Pi L \ket{\err_1} + 2\ket{\err_1}
\end{align*}
Therefore,
\[
    \norm{C \ket\init - \ket\target} \le \norm{2L\ket{\err_2}}  + \norm{4 L\Pi L \ket{\err_1}} + \norm{2 \ket{\err_1}} \le 2 \cdot \frac \epsilon 8 + 4 \cdot \frac \epsilon 8 + 2 \cdot  \frac \epsilon 8 = \epsilon.
\]

\end{proof}

\begin{theorem}\label{thm:HST_universal}
    Hadamard, $S$, and generalized Toffoli are universal for $\QAC^0$.
\end{theorem}
\begin{proof}
The proof follows almost immediately from \Cref{lem:real_approx}. First, write a single-qubit gate as $G=e^{i\alpha}V$, where $\alpha\in\mathbb R$ and $V\in\mathsf{SU}(2)$. The Euler decomposition gives the exact identity
\begin{equation}
\label{eq:euler}
    V= R_Y(\theta_1) R_X(\theta_2) R_Y(\theta_3) = R_Y(\theta_1) S^\dag R_Y(\theta_2) S R_Y(\theta_3)
\end{equation}
for angles $\theta_1, \theta_2, \theta_3 \in\mathbb R$, where $R_P(\theta)$ is the standard rotation gate about Pauli $P$: $R_P(\theta) := \cos(\theta/2) I - i \sin(\theta/2) P$. In particular, we have
\[
    R_Y(2 \theta) = \begin{pmatrix} \cos \theta & - \sin \theta \\ \sin \theta & \cos \theta \end{pmatrix} = \begin{pmatrix} \cos \theta & \sin \theta \\ \sin \theta & - \cos \theta \end{pmatrix} \begin{pmatrix} 1 & 0 \\ 0 & -1 \end{pmatrix} = U_\theta Z.
\]
Therefore, we can write \eqref{eq:euler} as
\[
    V= U_{\theta_1/2} Z S^\dag U_{\theta_2/2} Z S U_{\theta_3/2} Z = U_{\theta_1/2} S U_{\theta_2/2} S^3 U_{\theta_3/2} S^2
\]
so every determinant-one single-qubit gate can be approximated using three $U_\theta$ gates and six $S$ gates.

Consider a general $\QAC$ circuit $Q$ with $k\ge1$ single-qubit gates $G_j=e^{i\alpha_j}V_j$. Let $Q_0$ be the circuit obtained by replacing each $G_j$ with $V_j$, so that $Q=e^{i\alpha}Q_0$ for $\alpha=\sum_{j=1}^k\alpha_j$. To restore this phase, introduce one qubit initialized to $\ket0$ and apply
\[
    D_\alpha:=\begin{pmatrix}e^{i\alpha}&0\\0&e^{-i\alpha}\end{pmatrix}
    =H S^3 U_{-\alpha} S^3 H.
\]
Indeed, $S^3U_{-\alpha}S^3=\cos\alpha\,I+i\sin\alpha\,X$, and conjugating by $H$ gives the displayed diagonal matrix. Thus
\[
    (D_\alpha\otimes Q_0)\ket0\ket\psi=\ket0 Q\ket\psi,
\]
so the extra qubit returns to zero while supplying the required global phase.

To approximate $D_\alpha\otimes Q_0$ to error $\epsilon$, we use the above decomposition to obtain at most $3k+1$ gates of the form $U_\theta$, together with gates from the prescribed basis. We invoke \Cref{lem:real_approx} for each $U_\theta$ with error $\epsilon/(3k+1)$ and fresh work ancillas. By the triangle inequality, the resulting circuit is an $\epsilon$-approximation with $O(k \log^2(k/\epsilon))$ additional ancillary qubits.
\end{proof}

%% file: sections/exactification_1.tex
\section{Exact simulation of \texorpdfstring{$\TC^0$}{TC0} in \texorpdfstring{$\QAC^0$}{QAC0}}
\label{sec:exactification}
\label{sec:finite-spectrum-exactification}

In this section we show that the soundness error inherent in the $W$-test of
\cite{grier_morris_wu} can be eliminated entirely at the cost of additional
classical copies. Consequently, arbitrary threshold functions (and indeed
arbitrary $\TC^0$ functions) can be computed \emph{exactly} in $\QAC^0$ given
sufficiently many classical copies.

\subsection{Standard error reduction via parallel repetition }
Let us briefly recall the $W$-test and how it gives rise to a one-sided
Hamming weight tester:
\begin{theorem}[Corollary 19 of \cite{grier_morris_wu}]\label{thm:One-sided-tester}
    There is a family of $\QAC^0$ circuits
    $\{C_{n,k}\}_{n\geq 1,\,0\leq k\leq n}$ such that for all
    $x\in\{0,1\}^n$,
    \begin{itemize}
        \item If $|x|=k$, then $C_{n,k}$ accepts $\ket{x}$ with probability $1$.
        \item If $|x|\neq k$, then $C_{n,k}$ accepts $\ket{x}$ with probability at most $1-\frac{1}{n^2}$.
    \end{itemize}
\end{theorem}
The circuit witnessing \Cref{thm:One-sided-tester} in the case
$k=n/2$ is shown in \Cref{fig:w-test}.

\begin{figure}[htbp]
    \centering
    \input{figures/w-test}
    \caption{The $W$-test circuit, corresponding to the case of $k = n/2$}\label{fig:w-test}
\end{figure}
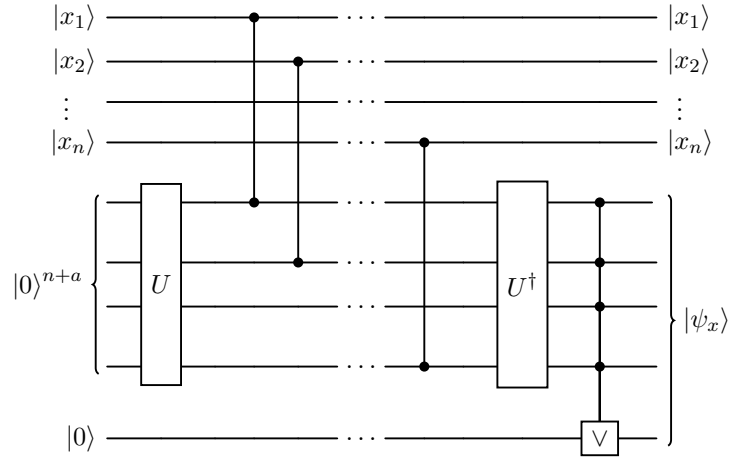

Here, $U$ is a constant-depth unitary that prepares
$\ket{W}=\frac{1}{\sqrt{n}}\sum_{i=1}^n\ket{e_i}$ with $a=a(n)$ clean ancilla qubits, and the two-qubit gates are $\CZ$ gates. The final $\OR$ acts on all $n+a$ workspace qubits and the output qubit. On input $\ket{x}$ the state evolves as follows:
\begin{align*}
    \ket{x}\ket{0^{n+a}}\ket{0}
    &\xrightarrow{U}\ket{x}\ket{W}\ket{0^a}\ket{0}\\
    &\xrightarrow{\CZ}\ket{x}\otimes\frac{1}{\sqrt{n}}
        \sum_{i=1}^n(-1)^{x_i}\ket{e_i}\ket{0^a}\ket{0}\\
    &\xrightarrow{U^{\dag}}\ket{x}
        \bigl(\alpha_x\ket{0^{n+a}}+\sqrt{p_1(x)}\ket{\psi_x}\bigr)\ket{0}\\
    &\xrightarrow{\OR}\ket{x}
        \bigl(\alpha_x\ket{0^{n+a}}\ket{0}
        +\sqrt{p_1(x)}\ket{\psi_x}\ket{1}\bigr)
\end{align*}
where $\ket{\psi_x}$ is a normalized state orthogonal to
$\ket{0^{n+a}}$, and
\begin{equation*}
    \alpha_x=\frac{n-2|x|}{n},\qquad
    p_0(x)=\alpha_x^2,\qquad p_1(x)=1-p_0(x).
\end{equation*}
$p_0(x)$ is exactly the probability that the output qubit yields $\ket{0}$ when measured in the computational basis. Thus, the test accepts with certainty when $|x|=n/2$, as $p_0(x) = 0$ and otherwise rejects
with probability at least $1/n^2$.

Given $m$ classical copies of $x$, we can run $m$ independent $W$-tests in
parallel and take the $\AND$ of their output bits $b_1,\ldots,b_m$. The probability that some test fails is
\begin{equation*}
    p_{\text{fail}} = 1-\bigl(1-p_0(x)\bigr)^m
        =1-\left(1-\frac{(n-2|x|)^2}{n^2}\right)^m.
\end{equation*}
If $|x|=n/2$, every test accepts. Otherwise, the probability that all tests
accept is at most $(1-1/n^2)^m\leq e^{-m/n^2}$, so parallel repetition
reduces the soundness error exponentially in $m/n^2$.

\subsection{Hamming weight oblivious exact amplitude amplification}
However, parallel repetition does not in general, eliminate this error: whenever
$0<p_1(x)<1$, the $W$ test fails with non-zero probability. We overcome this obstacle via an amplitude amplification protocol which reflects the post-$W$-test state entirely into the subspace in which the output qubit is in the correct state depending on $x$. Since $p_{\text{fail}} = \Omega(1)$ (see \Cref{fig:parallel_w_prob}) when $|x| \not = n/2$ this amplitude amplification can be done in a constant number of rounds and therefore constant depth - when $|x| = n/2$, $p_{\text{fail}} = 0$, so amplitude amplification will leave the state unchanged. 
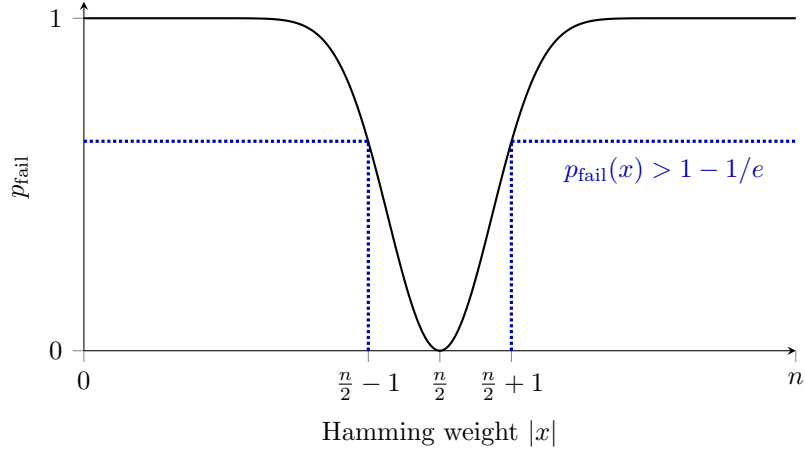
\begin{figure}[htbp]
    \centering
    \input{figures/graph}
    \caption{Hamming weight vs. acceptance probability of $\Omega(n^2)$ parallel $W$-tests.}
    \label{fig:parallel_w_prob}
\end{figure}

It may seem that this approach could not be made to work in $\QAC^0$ because applying precise amplitude amplification requires knowing exactly the angle between our starting state and the target state, but this overlap is a function of the Hamming weight of $x$ - the exact quantity we wish to compute. Nevertheless, we are able to construct a protocol which is able to make use of precise amplitude amplification while being oblivious to $|x|$.

 We first consider a promise problem wherein $|x|$ is either $k$ or $n/2$ for some fixed $k \not = n/2$. Then, we proceed in running some number, say $m_k$, parallel $W$-tests on $x$ and obtain a state $\ket{\psi_x}$. If $|x| = n/2$ then $\ket{\psi_x}$ will have zero overlap with the subspace corresponding to ``some W-tests fails", however if $|x| = k$ then we can choose $m_k$ such that this overlap is constant, i.e., $p_{\text{fail}} = \Omega(1)$. With an additional ancilla and an appropriate rotation applied, this overlap can be tuned to exactly $1/2$ \footnote{This technique is discussed at greater length in the proof of \cite[Corollary 9]{grier_morris_wu}.}. Finally, we proceed in applying amplitude amplification so that in the case of $|x| = k$ the resulting state is in the ``some W-tests fails" subspace and when $|x| = n/2$, the state is orthogonal to this subspace, allowing us to perfectly distinguish weight $k$ from weight $n/2$.
 
 If $|x| = j \not \in \{k, n/2\}$ then this procedure will not produce a meaningful result. However, if we had instead tested $j$ vs. $n/2$, then we would have been able to correctly conclude $|x| \not = n/2$. So, with sufficiently many classical copies of $x$ we may run the $k$ vs. $n/2$ test on $x$ for each $k \in \{0, 1, \dots n\} \setminus {n/2}$ and if $|x| \not = n/2$, then \emph{some} test will fail with certainty. Taking the $\AND$ of the outputs of each of these sub-tests will give $\EX_{n/2}(x)$ exactly. By choosing the number of copies separately for each $k$, we can solve all $n$ promise problems in parallel using $O(n^2)$ classical copies in total.

Let us now construct the circuit which perfectly  distinguishes the middle Hamming slice from any other fixed Hamming slice:
\begin{samepage}
\begin{lemma}\label{lem:single-slice}
    For every even $n\geq 2$ and every $k\in\{0,\ldots,n\}\setminus\{n/2\}$,
    there is a $\QAC^0$ circuit $Q_k$ using $m_k=O(n^2)$ classical copies
    of $x\in\{0,1\}^n$ and polynomially many zero-initialized ancillas such that
    \begin{itemize}
        \item $Q_k$ accepts $\ket{x}^{\otimes m_k}$ with certainty if $|x|=n/2$.
        \item $Q_k$ rejects $\ket{x}^{\otimes m_k}$ with certainty if $|x|=k$.
    \end{itemize}
\end{lemma}
\end{samepage}
\begin{proof}
    Fix $k\neq n/2$ and set
    \begin{equation*}
        \Delta:=|n/2-k|,\qquad
        m:=m_k=\left\lceil\frac{n^2}{4\Delta^2}\right\rceil.
    \end{equation*}
    Run one $W$-test on each of the $m$ input copies in parallel. Let $(b_1, \dots b_m) \in \{0, 1\}^m$ denote the results of each of these tests. On weight $k$, a single test rejects with probability
    $4\Delta^2/n^2$, so the probability of at least one rejection is
    \begin{equation*}
        q_k:=1-\left(1-\frac{4\Delta^2}{n^2}\right)^m
        \geq 1-e^{-1}.
    \end{equation*}
    This value depends only on $n$ and $k$. In contrast, on weight $n/2$
    the rejection probability is zero. Below, $p_{\text{fail}}$ denotes the rejection
    probability for this choice of $m$.

    Compute the $\AND$ of the output bits into a new ancilla initialized to $\ket{0}$ and then apply $X$ to that qubit, so that its value becomes
    $r=1-(b_1\wedge\cdots\wedge b_m)$. We will refer to this register as the $\AND$ register. The resulting state is
    \begin{equation*}
        \ket{\gamma_x}
        =\sqrt{p_{\text{fail}}}\ket{\gamma_0^x}\ket{1}
        +\sqrt{1-p_{\text{fail}}}\ket{\gamma_1^x}\ket{0}.
    \end{equation*}
    Where $\ket{\gamma_0^x}$ is a normalized state supported only on basis states satisfying $b_1\wedge\cdots\wedge b_m = 0$ and $\ket{\gamma_1^x}$ on basis states where $b_1\wedge\cdots\wedge b_m = 1$.

    For $\theta\in[0,\pi/2]$, define the controlled one-qubit gate
    \begin{equation*}
        \mathsf{C}U(\theta) = \mathsf{C}\bigl(R_Y(2\theta)Z\bigr)
        :=\proj{1}\otimes
        \begin{pmatrix}
            \cos{\theta} & \sin{\theta}\\
            \sin{\theta} & -\cos{\theta}
        \end{pmatrix}
        +\proj{0}\otimes\mathbb{I}.
    \end{equation*}
    Append an ancilla qubit initialized to $\ket{0}$ and apply
    $\mathsf{C}U(\theta_k)$ controlled on the $\AND$ register, where
    \begin{equation*}
        \cos^2{\theta_k}:= t_k = 1-\frac{1}{4q_k}\in(0,3/4]
    \end{equation*}
    The resulting state is
    \begin{align*}
        \ket{\Gamma_x}
        &=\sqrt{t_kp_{\text{fail}}}\ket{\gamma_0^x}\ket{10}
          +\sqrt{(1-t_k)p_{\text{fail}}}\ket{\gamma_0^x}\ket{11}\\
        &\qquad+\sqrt{1-p_{\text{fail}}}\ket{\gamma_1^x}\ket{00}.
    \end{align*}
    Note that if $|x| = n/2$ then $\ket{\Gamma_x} = \ket{\gamma_1^x}\ket{00}$, but if $|x| = k$ then $(\bra{\gamma_0^x}\bra{11})\ket{\Gamma_x} = 1/2$.
    
    Now, we apply $\CZ$ to the last two qubits of this state:
    \begin{equation*}
        \ket{\Gamma_x^-}
        :=(\mathbb{I}\otimes\CZ)\ket{\Gamma_x}
        =\ket{\Gamma_x}
          -2\sqrt{(1-t_k)p_{\text{fail}}}\ket{\gamma_0^x}\ket{11}.
    \end{equation*}
    Next, apply $R_x:=\mathbb{I}-2\proj{\Gamma_x}$. This reflection can be implemented via
    
    \begin{equation*}
    R_x=V_{k,x}\bigl(\mathbb{I}-2\proj{0^B}\bigr)V_{k,x}^{\dag}.
    \end{equation*}
    where $V_{k, x}$ is the unitary that prepares $\ket{\Gamma_x}$ from $\ket{x}\otimes\ket{0^B}$. Note that $V_{k, x}$ only controls on the copies of $x$, so implementing $V_{k, x}$ requires no knowledge of $x$ or its Hamming weight.

    If $|x|=k$, then $(1-t_k)p_{\text{fail}}=1/4$ and
    \begin{equation*}
        \braket{\Gamma_x}{\Gamma_x^-}
        =1-2(1-t_k)p_{\text{fail}}=\frac{1}{2}.
    \end{equation*}
    Therefore,
    \begin{align*}
        R_x\ket{\Gamma_x^-}
        &=\ket{\Gamma_x^-}
          -2\braket{\Gamma_x}{\Gamma_x^-}\ket{\Gamma_x}\\
        &=\ket{\Gamma_x^-}-\ket{\Gamma_x}\\
        &=-2\sqrt{(1-t_k)p_{\text{fail}}}\ket{\gamma_0^x}\ket{11}\\
        &=-\ket{\gamma_0^x}\ket{11}
    \end{align*}
    If $|x|=n/2$, then $p_{\text{fail}}=0$, so the $\CZ$ gate has no effect
    and reflection yields
    \begin{equation*}
        R_x\ket{\Gamma_x^-}=-\ket{\Gamma_x}
        =-\ket{\gamma_1^x}\ket{00}
    \end{equation*}
    So, when $|x| = n/2$ the state is unchanged up to a global phase. Finally, apply $X$ to the final ancilla qubit, so that it is $1$ when $|x|=n/2$ and $0$ when $|x|=k$. Thus, if final ancilla qubit is taken as the output then this circuit solves the $k$ vs. $n/2$ Hamming weight problem with zero-error on all inputs satisfying the promise.
\end{proof}

To achieve a more general test in which weight $j \not = n/2$ strings are accepted
and weight $k$ strings are rejected, we can append the fixed string $1^{n-j}0^j$
to each input copy. The resulting $2n$-bit string
$y=x1^{n-j}0^j$ satisfies $|y| = n$ when $|x|=j$, and has weight $n+k-j$ when
$|x|=k$. Applying \Cref{lem:single-slice} for inputs of length $2n$ gives the desired
behavior with $O(n^2)$ copies of $x$. The padded bits can be hard-coded into the circuit with $X$ gates.

\begin{theorem}\label{thm:zero_error_ex}
    For even $n$, $\EX_{n/2}$ can be computed with zero error in $\QAC^0$
    given $m=O(n^2)$ classical copies of its input. In particular,
    $\EX_{n/2}\in\EQAC^0\circ\NC^0$, and the corresponding lifted function
    $\EX_{n/2}^{\uparrow m}$ is in $\EQAC^0$.
\end{theorem}
\begin{proof}
    Run the circuit $Q_k$ from \Cref{lem:single-slice} for every
    $k\in\{0,\ldots,n\}\setminus\{n/2\}$ in parallel, using disjoint
    input copies. Taking the $\AND$ of each of the output bits of each $Q_k$ computes $\EX_{n/2}$ exactly using $O(n^2)$ copies of the input.
    Indeed, the total number of copies is
    \[
        \sum_{k\ne n/2}m_k
        =2\sum_{\Delta=1}^{n/2}\left\lceil\frac{n^2}{4\Delta^2}\right\rceil
        \leq n+\frac{n^2}{2}\sum_{\Delta=1}^{n/2}\frac{1}{\Delta^2}
        =O(n^2).
    \]
    To compute the lifted function on arbitrary blocks $x^{(1)},\ldots,x^{(m)}$, first check that all blocks agree. For each coordinate $i\in[n]$, test whether the $m$ bits $x_i^{(1)},\ldots,x_i^{(m)}$ are all zero or all one. These tests use generalized Toffoli and single-qubit gates, run in parallel across coordinates, and preserve the input blocks. Taking the $\AND$ of their results gives a consistency flag in constant depth and $O(nm)$ size. Run the above circuit on the blocks and take the $\AND$ of its output with this flag. Equal blocks give the correct value of $\EX_{n/2}$, while unequal blocks are rejected with certainty. This proves $\EX_{n/2}^{\uparrow m}\in\EQAC^0$.
\end{proof}

The same padding argument gives exact tests for every $\EX_j$. Taking an
$\OR$ over the appropriate Hamming weights then computes any symmetric function,
including every threshold function, using polynomially many classical
copies. As in \cite[Section~4.2]{grier_morris_wu}:
\begin{corollary}
    $\TC^0\subseteq\EQAC^0\circ\NC^0$.
\end{corollary}
\begin{corollary}\label{cor:exact_separation}
    For all primes $p$, $\EQAC^0 \not \subseteq \AC^0[p]$. Additionally, $\EQAC^0 \not \subseteq \mathsf{ACC}^0$ under the assumption that $\mathsf{ACC}^0 \subset \TC^0$.
\end{corollary}
\begin{proof}
    For a fixed prime $p$, if there exist a constant $k > 0$ such that $\EX_{n/2}^{\uparrow n^k} \in \AC^0[p]$, then $\EX_{n/2}\in \AC^0[p]$. To see this, suppose $C$ is the polynomial-size $\AC^0[p]$ circuit which computes $\EX_{n/2}^{\uparrow n^k}$. We can construct an $\AC^0[p]$ circuit $\Tilde{C}$ which computes $\EX_{n/2}$ by first fanning the input bits out $n^k$ times and then evaluating $C$ on the resulting string, $\overbrace{x \circ \cdots \circ x}^{n^k} \in \{0, 1\}^{n^{k + 1}}$:
    \begin{figure}[ht]
    \centering
    \input{figures/classical_fanout_reduction}
    \caption{If $C$ computes $\EX_{n/2}^{\uparrow n^k}$ then we can construct a circuit $\Tilde{C}$ which computes $\EX_{n/2}$ by simply copying the input and feeding it into $C$.}
    \label{fig:fanout_reduction}
\end{figure}
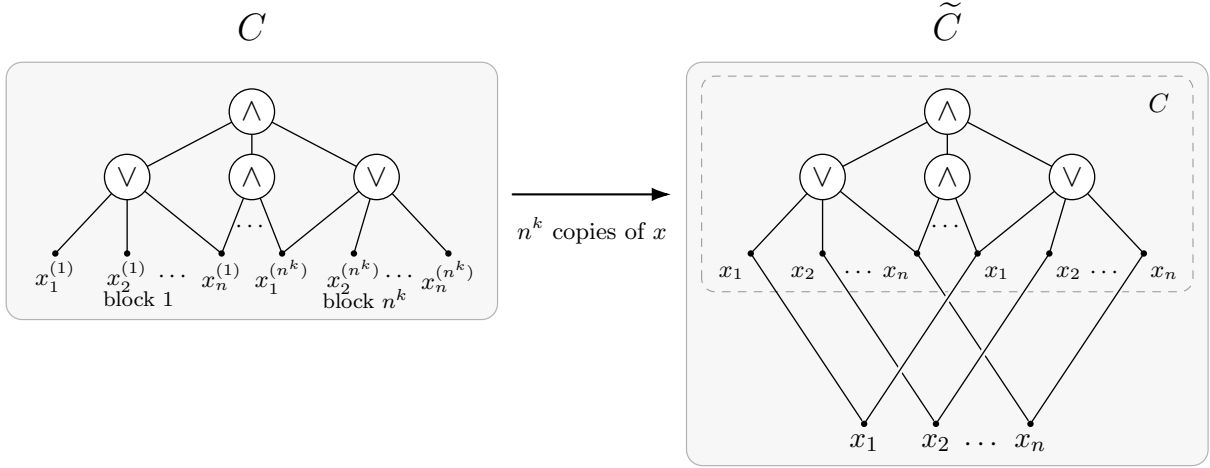
    However, $\EX_{n/2}\in \AC^0[p] \implies \AC^0[p] = \TC^0$ - which we know to be false via \cite{razborov, smolensky}. Similarly, $\EX_{n/2}\in \mathsf{ACC}^0 \implies \mathsf{ACC}^0 = \TC^0$.
\end{proof}

%% file: figures/w-test.tex
\scalebox{0.9}{
    \begin{quantikz}
    \lstick{\ket{x_1}}  & & & \control{} &  & \ \ldots\  & & & & & \rstick{\ket{x_1}}\\
    \lstick{\ket{x_2}} & & & & \control{} & \ \ldots\  & & & & &  \rstick{\ket{x_2}}\\
    \lstick{\vdots \ \ \ } & & & & & \ \ldots \  & & & & &\rstick{\ \vdots \ \ }\\
    \lstick{\ket{x_n}}  & & & & & \ \ldots\  & \control{} & & & &\rstick{\ket{x_n}}\\
    \lstick[wires=4]{$\ket{0}^{n + a}$} & \gate[wires=4]{U} & &  \ctrl{-4} & & \ \ldots \ & & & \gate[wires=4]{U^{\dag}} &\ctrl{4} & \ \rstick[wires=5]{\ \ket{\psi_{x}}}\\ 
    & & &  & \ctrl{-4} & \ \ldots \  & & & & \ctrl{3} &\\
     & & & & & \ \ldots \ & & & & \ctrl{2} &\\
     & & & & & \ \ldots \ & \ctrl{-4} & & & \ctrl{1}&\\
     \lstick{\ket{0}}&  & & & & \ \ldots \ & & & & \gate{\lor} &
  \end{quantikz}
  }

%% file: figures/graph.tex
\begin{tikzpicture}
  \def\N{20}
  \pgfmathsetmacro{\Mid}{\N/2}
  \pgfmathsetmacro{\Constant}{.63}
  % Smaller powers give a broader central dip and more label space.
  \def\CurvePower{24}
  % Intersections of the schematic curve with y=\Constant.
  \pgfmathsetmacro{\Offset}{\Mid*sqrt(1-pow(1-\Constant,1/\CurvePower))}
  \pgfmathsetmacro{\NearLeft}{\Mid-\Offset}
  \pgfmathsetmacro{\NearRight}{\Mid+\Offset}
  \pgfmathsetmacro{\LabelX}{0.97*\N}

  \begin{axis}[
      width=11cm,
      height=6.2cm,
      xmin=0, xmax=\N,
      ymin=0, ymax=1.05,
      axis lines=left,
      xlabel={Hamming weight $|x|$},
      ylabel={$p_{\text{fail}}$},
      % Symbolic labels with schematic horizontal spacing.
      xtick={0,\NearLeft,\Mid,\NearRight,\N},
      xticklabels={$0$,$\frac{n}{2}-1$,$\frac{n}{2}$,$\frac{n}{2}+1$,$n$},
      ytick={0,1},
      tick label style={font=\small},
      label style={font=\small},
      tick align=outside,
      clip=false
    ]

    % Broader schematic curve for a less crowded x-axis.
    \addplot[
      black,
      thick,
      no marks,
      domain=0:\N,
      samples=601
    ]
    {1-pow(4*x*(\N-x)/(\N*\N),\CurvePower)};

    % Keep the original horizontal level; meet the curve at both ends.
    \addplot[
      blue!70!black,
      densely dotted,
      very thick,
      no marks
    ] coordinates {(0,\Constant) (\NearLeft,\Constant)};
    \addplot[
      blue!70!black,
      densely dotted,
      very thick,
      no marks,
      forget plot
    ] coordinates {(\NearRight,\Constant) (\N,\Constant)};

    % Vertical guides connect the ticks to the intersections.
    \addplot[
      blue!70!black,
      densely dotted,
      very thick,
      no marks,
      forget plot
    ] coordinates {(\NearLeft,0) (\NearLeft,\Constant)};
    \addplot[
      blue!70!black,
      densely dotted,
      very thick,
      no marks,
      forget plot
    ] coordinates {(\NearRight,0) (\NearRight,\Constant)};

    \node[
      anchor=north east,
      font=\small,
      text=blue!70!black,
      yshift=-2pt
    ] at (axis cs:\LabelX,\Constant)
      {$p_{\text{fail}}(x) > 1 - 1/e$};
  \end{axis}
\end{tikzpicture}

%% file: figures/classical_fanout_reduction.tex
% Reduction C -> \widetilde C in the proof of Corollary 20.
% Both panels contain the same schematic gate network and input ports.
% A horizontal row of inputs fans out diagonally into the n^k blocks of C;
% the first and last blocks and bits 1, 2, and n are shown explicitly.
% Include with \input{figures/classical_fanout_reduction}.
\begingroup
\usetikzlibrary{arrows.meta}
\begin{tikzpicture}[
    x=1cm, y=0.92cm,
    line width=0.5pt, line cap=round, line join=round,
    panel/.style={draw=black!30, fill=black!3, rounded corners=6pt,
                  line width=0.45pt},
    gate/.style={circle, draw, fill=white, minimum size=6mm,
                 inner sep=0pt, font=\large},
    pin/.style={font=\footnotesize, inner sep=0pt},
    block/.style={font=\scriptsize, inner sep=1pt},
    every node/.style={outer sep=0pt}
]
    \begin{scope}[xshift=-4.60cm]
        \path[panel] (-3.20,-0.15) rectangle (3.30,3.55);
        \node[font=\Large, anchor=base] at (0.05,3.87) {$C$};
    \end{scope}
    \begin{scope}[xshift=4.60cm]
        \path[panel] (-3.40,-2.25) rectangle (3.50,3.55);
        \draw[draw=black!40, dashed, rounded corners=5pt, line width=0.4pt]
            (-3.20,0.25) rectangle (3.30,3.35);
        \node[font=\Large, anchor=base] at (0.05,3.87) {$\widetilde C$};
        \node[font=\small] at (2.85,2.97) {$C$};

        % Each source forks directly into its matching input in both blocks.
        % Parallel diagonal bundles avoid right-angle bends. Background gaps
        % distinguish crossings from the three genuine source junctions.
        \foreach \xx in {-1.05,-0.10,1.15}
            \draw (\xx,-1.65) -- ({\xx-1.50},0.80);
        \foreach \xx in {-1.05,-0.10,1.15}
            \draw[preaction={draw=black!3, line width=2.5pt}]
                (\xx,-1.65) -- ({\xx+1.50},0.80);
        \foreach \xx/\bit in {-1.05/1,-0.10/2,1.15/n} {
            \fill (\xx,-1.65) circle (1.2pt);
            \node[below=4pt, inner sep=0pt] at (\xx,-1.65) {$x_{\bit}$};
        }
        \node at (0.525,-1.94) {$\cdots$};
    \end{scope}

    % The gate network is drawn only once in the source and used in both panels.
    % The central AND joins different input blocks of the same C.
    \foreach \side/\shift in {left/-4.60,right/4.60} {
        \begin{scope}[xshift=\shift cm]
            \coordinate (\side-a) at (-2.55,0.80);
            \coordinate (\side-b) at (-1.60,0.80);
            \coordinate (\side-c) at (-0.35,0.80);
            \coordinate (\side-d) at (0.45,0.80);
            \coordinate (\side-e) at (1.40,0.80);
            \coordinate (\side-f) at (2.65,0.80);
            \coordinate (\side-or1) at (-1.60,1.90);
            \coordinate (\side-and) at (0.05,1.90);
            \coordinate (\side-or2) at (1.70,1.90);
            \coordinate (\side-out) at (0.05,2.84);
            \draw (\side-a) -- (\side-or1) -- (\side-out);
            \draw (\side-b) -- (\side-or1);
            \draw (\side-c) -- (\side-or1);
            \draw (\side-c) -- (\side-and) -- (\side-out);
            \draw (\side-d) -- (\side-and);
            \draw (\side-d) -- (\side-or2) -- (\side-out);
            \draw (\side-e) -- (\side-or2);
            \draw (\side-f) -- (\side-or2);
            \foreach \port in {a,b,c,d,e,f}
                \fill (\side-\port) circle (1.2pt);
            \node[gate] at (\side-or1) {$\vee$};
            \node[gate] at (\side-and) {$\wedge$};
            \node[gate] at (\side-or2) {$\vee$};
            \node[gate] at (\side-out) {$\wedge$};
        \end{scope}
    }

    \begin{scope}[xshift=-4.60cm]
        \node[pin] at (-2.55,0.47) {$x_1^{(1)}$};
        \node[pin] at (-1.60,0.47) {$x_2^{(1)}$};
        \node[pin] at (-0.35,0.47) {$x_n^{(1)}$};
        \node[pin] at (0.45,0.47) {$x_1^{(n^k)}$};
        \node[pin] at (1.40,0.47) {$x_2^{(n^k)}$};
        \node[pin] at (2.65,0.47) {$x_n^{(n^k)}$};
        \node[font=\footnotesize] at (-0.975,0.47) {$\cdots$};
        \node[font=\footnotesize] at (2.025,0.47) {$\cdots$};
        \node[font=\footnotesize] at (0.05,1.18) {$\cdots$};
        \node[block] at (-1.45,0.15) {block $1$};
        \node[block] at (1.55,0.15) {block $n^k$};
    \end{scope}
    \begin{scope}[xshift=4.60cm]
        % Labels sit outside each diagonal bundle so the wires remain visible.
        \node[pin, anchor=east] at (-2.64,0.52) {$x_1$};
        \node[pin, anchor=east] at (-1.69,0.52) {$x_2$};
        \node[pin, anchor=east] at (-0.44,0.52) {$x_n$};
        \node[pin, anchor=west] at (0.54,0.52) {$x_1$};
        \node[pin, anchor=west] at (1.49,0.52) {$x_2$};
        \node[pin, anchor=west] at (2.74,0.52) {$x_n$};
        \node[font=\footnotesize] at (-1.10,0.52) {$\cdots$};
        \node[font=\footnotesize] at (2.15,0.52) {$\cdots$};
        \node[font=\footnotesize] at (0.05,1.18) {$\cdots$};
    \end{scope}

    \draw[-{Latex[length=2.6mm,width=1.8mm]}, line width=0.8pt]
        (-1.10,1.65) -- (1.00,1.65)
        node[midway, above=5pt, font=\small] {}
        node[midway, below=5pt, font=\footnotesize] {$n^k$ copies of $x$};
\end{tikzpicture}
\endgroup

%% file: sections/appendices.tex
\section{Random selectors in \texorpdfstring{$\QAC^0$}{QAC0}}\label{sec:sampling}

Given a binary string $x\in\{0,1\}^n$, the selector function checks if $x$ contains $1$ and outputs one such location. 
In this section, we give two constructions in $\QAC^0$: in the binary version, it outputs a random $i\in[n]$ with $x_i=1$; and in the unary version, it outputs the indicator vector $e_i=0^{i-1}10^{n-i}$.

We remark that if $\QAC^0$ could implement the map $\ket i\to\ket{e_i}$, then the two constructions would be equivalent. Unfortunately this operator is equivalent to computing the indexing function, which is not known to be in $\QAC^0$. To this this, recall that the indexing function computes $x_i$ given input $i$ and $x$; if the operator is available, one can obtain $e_i$ from $i$ and then an entrywise-AND followed by an OR suffices. Consequently we view our selector constructions quite surprising: it bypasses the natural need of computing the indexing function.

\paragraph{The unary selector.}
Write $S(x)=\set{i:x_i=1}$ for the support of $x$.  We first construct the unary selector, whose circuit aims to return $e_i$ for a uniformly random $i\in S(x)$.  
The main challenge is to produce all $n$ output bits of $e_i$ without large fanout.  We do this by the phase kickback trick.

\begin{theorem}\label{thm:one-hot-selector}
Let $\delta\in(0,1/2)$ satisfy $\log(1/\delta)=\polylog(n)$.  There is a constant-depth, polynomial-size $\QAC$ circuit with an $n$-qubit output register $Y$ such that:
\begin{enumerate}
\item if input $x=0^n$, then measuring $Y$ gives $0^n$ with probability one;
\item if input $x\neq0^n$, then measuring $Y$ gives either $0^n$ or $e_i$ where $i$ is uniformly random in $S(x)$. 
In addition, with probability at least $1-\delta$, the latter case happens.
\end{enumerate}
\end{theorem}
\begin{proof}
Set
\[
    R=\lceil\log_2(n)\rceil,
    \qquad q=\lceil 2e\ln(1/\delta)\rceil,
    \qquad T=(R+1)q=\polylog(n).
\]
For each trial $t=(r,j)$ with $0\leq r\leq R$ and $1\leq j\leq q$, choose a mask $m_t\in\{0,1\}^n$ by including each bit independently with probability $2^{-r}$.
Define the masked vector and its success flag by
\[
    z_t=x\wedge m_t
    \quad\text{and}\quad
    u_t=\EX_1(z_t).
\]
Fix an arbitrary ordering of the trials and let
\[
    s_t=u_t\wedge\bigwedge_{t'<t}(\neg u_{t'}),
\]
which marks the first trial that isolates exactly one coordinate of $S(x)$. Thus $s$ is either all-zero or an indicator vector with exactly one $1$. The desired output is
\[
    F(x)=
    \begin{cases}
        z_t&\text{if }s_t=1\text{ for some }t,\\
        0^n&\text{if every trial fails}.
    \end{cases}
\]
In particular, $F(x)$ is always zero or a valid indicator vector, and it is zero whenever $x=0^n$.

Suppose $k=|S(x)|>0$.  If $k=1$, every trial at level $r=0$ succeeds.  Otherwise, at level $r=\lceil\log_2 k\rceil$, the sampling probability $2^{-r}$ lies in $[1/(2k),1/k]$, and the trial succeeds with probability
\[
    k\cdot2^{-r}\left(1-2^{-r}\right)^{k-1}\geq\frac1{2e}.
\]
Consequently, the probability that all trials fail is at most $\left(1-\frac1{2e}\right)^q\leq e^{-q/(2e)}\leq\delta$.
Moreover, within any trial, all $k$ marked coordinates have the same probability of being the unique survivor.  For every trial with positive success probability, the survivor conditional on success is therefore uniform on $S(x)$.
This proves both the success bound and the claimed conditional uniformity for the random function $F$.

We now implement this procedure coherently in $\QAC^0$.  
Exact polylogarithmic fanout~\cite{grier_morris_wu} supplies the $T$ copies of each $x_i$ needed to compute all $z_t$ in parallel.  The exact $\EX_1$ construction~\cite{grier_morris_wu} computes the flags $u_t$ in parallel. Each bit of $u$ is used at most $T$ times in among $s_t$'s, so another application of exact polylogarithmic fanout allows all $s_t$'s to be computed in parallel.
This computes all $z_t$'s and $s_t$'s.

Recall that $F(x)$ depends only on $z_t$'s and $s_t$'s.
To compute $F(x)$, we implement phase oracle
\[
    P_F:\ket{z_{t}\text{'s}}\ket{s_t\text{'s}}\ket a
       \longmapsto(-1)^{\langle a,F(x)\rangle}\ket{z_t\text{'s}}\ket{s_t\text{'s}}\ket a
       \quad\text{for $a\in\{0,1\}^n$,}
\]
where we omit other registers for simplicity.
To this end, for $a\in\{0,1\}^n$ and each $t$, we compute $d_t=\OR_{i=1}^n(a_i\wedge z_{t,i})$ and let $h=\OR_t(s_t\wedge d_t)$.
If $F(x)=z_t$, the vector $z_t$ has exactly one nonzero coordinate, so $d_t=\langle a,z_t\rangle$; otherwise $F(x)=0^n$, the $s$ vector is all-zero, so $h=0$. 
In either case, we have $h=\langle a,F(x)\rangle$, after which a bit-phase oracle $\ket h\to(-1)^h\ket h$, together with appropriate uncomputation, constructs $P_F$.

Finally, to write $F(x)$ on register $Y$, we use phase kickback based on $P_F$: initialize $Y$ register to $0^n$ and apply $H^{\otimes n}$, $P_F$ with $a=Y$, and $H^{\otimes n}$ again, which are all in $\QAC^0$. This completes the whole proof.
\end{proof}

\paragraph{The binary selector.}
For the binary selector, we fix
$$
L=\lceil\log_2(n+1)\rceil.
$$
Since each index $i$ has only $L$ bits, we can compute those bits directly from its masked vector using polylogarithmic fanout.
Below, we use $\operatorname{bin}(i)$ to denote the length-$L$ binary encoding of $i\in[n]$.

\begin{theorem}\label{thm:binary-selector}\label{app:binary-selector}
Let $\delta\in(0,1/2)$ satisfy $\log(1/\delta)=\polylog(n)$.
There is a constant-depth, polynomial-size $\QAC$ circuit with an $L$-qubit output register $Y$ such that
\begin{enumerate} 
\item If input $x=0^n$, then measuring $Y$ gives $0^L$ with probability one;
\item if input $x\ne0^n$, then measuring $Y$ gives either $0^L$ or $\operatorname{bin}(i)$ where $i$ is uniformly random in $S(x)$.
In addition, with probability at least $1-\delta$, the latter case happens.
\end{enumerate}
\end{theorem}
\begin{proof}
Given the unary selector, the binary selector is much easier: we simply observe that the map $\ket{e_i}\to\ket{e_i}\ket{\operatorname{bin}(i)}$ is in $\QAC^0$. This is because there are only $L$ bits in $\operatorname{bin}(i)$ and, since $e_i$ has Hamming weight $1$, each such bit can be computed by an OR of subsets of bits in $e_i$. Consequently (poly)logarithmic fanout suffices.
\end{proof}

We make a final remark there.
While the map $\ket{e_i}\to\ket{e_i}\ket{\operatorname{bin}(i)}$ is in $\QAC^0$, as used in the above proof, we do not know if it entails $\ket{e_i}\to\ket i$ in $\QAC^0$, which would put the indexing function into $\QAC^0$.

\section{Approximate counting in \texorpdfstring{$\QAC^0$}{QAC0}}

In this section, we show how $\QAC^0$ can count, up to small multiplicative error.
Our construction is based on the $\mathsf{AC}^0$-style counting sketch~\cite{ajtai_approximate_counting,viola_approximate_counting}, which requires only polylogarithmic fanout and small threshold functions.

\begin{theorem}\label{thm:approx-counter}
Let $\eps,\delta\in(0,1/2)$ satisfy $1/\eps=\polylog(n)$ and $\log(1/\delta)=\polylog(n)$.
There is a constant-depth, polynomial-size $\QAC$ circuit that, on input $x\in\{0,1\}^n$, outputs an integer $w\in\set{0,\ldots,n}$ such that
\begin{enumerate}
\item if Hamming weight $|x|=0$, then $w=0$ with probability one;
\item if Hamming weight $|x|>0$, then $(1-\eps)\abs{x}\leq w\leq(1+\eps)\abs{x}$ with probability at least $1-\delta$.
\end{enumerate}
\end{theorem}
\begin{proof}
The proof combines the classic geometric subsampling sketch~\cite{ajtai_approximate_counting,viola_approximate_counting} with the known polylogarithmic fanout and threshold function constructions in $\QAC^0$ \cite{grier_morris_wu}.

We briefly outline the classical approximate counting algorithm.
Set $B=\Theta(\eps^{-2})$ and use $J=O(\eps^{-1}\log n)$ candidate weights $T_j=\lfloor B\cdot(1+\Theta(\eps))^j\rceil$ extending up to $n$.
For each Hamming weight $w$ at most $2B$, we use threshold functions to directly check if $|x|=w$.
For each $j\in[J]$, we perform the following subsampling algorithm: independently keep each coordinate of $x$ with probability $B/T_j$ and check whether at least $B$ kept coordinates are ones; note that the expected number of ones is $B|x|/T_j$.
Let $j'\in[J]$ be the largest candidate that succeeds the test.
If $|x|\le2B$, we have directly computed $|x|$; otherwise, we use $T_{j'}$ as the approximation of $|x|$.
Standard concentration bounds show that the multiplicative error is at most $\eps$ with constant success probability. Then the ultimate construction computes the median of $O(\log(1/\delta))$ independent runs to reduce the failure probability to $\delta$.
The detailed analysis can be found in \cite{ajtai_approximate_counting,viola_approximate_counting}.

To implement the above algorithm in $\QAC^0$, we recall from \cite{grier_morris_wu} that threshold functions of polylogarithmic weights and polylogarithmic fanout have exact implementation in $\QAC^0$. By our assumption on $\delta$ and $\eps$, one can implement the $O(\log(1/\delta))$ runs exactly in $\QAC^0$.
As for the final median computation, we observe that it is the median of polylogarithmic many numbers in $S:=\{0,1,\ldots,2B\}\cup\{T_1,\ldots,T_J\}$. Since both $B$ and $J$ are polylogarithmic, we can check individually if each $s\in S$ is the median. This individual check can then be implemented as a polylogarithmic-sized majority of greater-than functions on two binary strings of polylogarithmic length. Hence these also only use threshold functions of polylogarithmic weights and polylogarithmic fanout. This completes the whole proof.
\end{proof}